\documentclass[a4paper,UKenglish, autoref,thm-restate]{lipics-v2021}

\ccsdesc[500]{Theory of computation~Theory and algorithms for application domains}
\ccsdesc[500]{Theory of computation~Approximation algorithms analysis}
\ccsdesc[500]{Networks~Network algorithms}
\Copyright{Stefan Schmid, Jakub Svoboda, Michelle Yeo}
\nolinenumbers

\usepackage{amsmath}
\usepackage{amsthm}
\usepackage{csquotes}
\usepackage{booktabs}

\usepackage{pgfplots}
\pgfplotsset{compat=1.15}
\usepackage{tikz}
\usetikzlibrary{shapes,positioning,arrows,arrows.meta,calc,automata,matrix,fit,backgrounds, automata}
\tikzstyle{state}+=[minimum size = 6mm, inner sep=0,outer sep=1]

\colorlet{disabled}{lightgray}
\tikzstyle{state}=[draw,rectangle,inner sep=5pt,rounded corners=2pt]
\tikzstyle{action}=[font=\small,inner sep=0pt,outer sep=3pt]
\tikzstyle{actionnode}=[circle,draw=black,fill=black,minimum size=1mm,inner sep=0,outer sep=0]
\tikzstyle{actionedge}=[draw,-]
\tikzstyle{prob}=[font=\scriptsize,inner sep=0pt,outer sep=1pt]
\tikzstyle{probedge}=[draw,->]
\tikzstyle{directedge}=[draw,->]
\tikzset{chainarrow/.tip={Stealth[length=3pt]}}
\tikzset{>=chainarrow}

\usepackage{xparse} 
\usepackage{mathtools}
\usepackage{environ}
\usepackage{marvosym}
\usepackage{bbm}
\usepackage{framed}

\usepackage{algorithmicx,algorithm}
\usepackage[noend]{algpseudocode}

\usepackage[capitalize]{cleveref}

\usepackage{enumitem}
\usepackage{subcaption}
\usepackage{thm-restate}

\newcommand{\eps}{\varepsilon}

\newcommand{\xrl}{X_{\leftarrow}}
\newcommand{\xlr}{X_{\rightarrow}}
\newcommand{\opt}{OPT }
\newcommand{\accAll}{M_{\max}}
\newcommand{\problem}{weighted packet selection for a link }
\newcommand{\Problem}{Weighted packet selection for a link }
\newcommand{\capacitySum}{M}

\newcommand{\surre}{\phi_R}
\newcommand{\surac}{\phi_A}
\newcommand{\undecided}{U}

\newcommand{\arat}{1 + \sqrt{3}} 
\newcommand{\hibu}{\sqrt{3}-1} 

\newcommand*\calO{\mathcal{O}}

\usepackage{graphicx}
\graphicspath{ {./images/} }

\usepackage{url}
\title{Weighted Packet Selection for Rechargeable Links: Complexity and Approximation}

\author{Stefan Schmid}{Technische Universität Berlin, Germany}{}{}{Research supported by the European Research Council (ERC), grant agreement No. 864228
(AdjustNet), Horizon 2020, 2020-2025}{}{}
\author{Jakub Svoboda}{Institute of Science and Technology, Austria}{}{}{}{}{}{}
\author{Michelle Yeo}{Institute of Science and Technology, Austria}{}{}{}{}{}{}
\date{April 2022}
\hideLIPIcs

\authorrunning{S. Schmid, J. Svoboda, and M. Yeo}
\titlerunning{Packet Selection for Rechargeable Links}

\keywords{network algorithms, approximation algorithms, complexity, cryptocurrencies} 
\begin{document}

\maketitle
\begin{abstract}
We consider a natural problem dealing with weighted packet selection across a rechargeable link, which e.g., finds applications in cryptocurrency networks.
The capacity of a link $(u,v)$ is determined by how much players $u$ and $v$ allocate for this link. 
Specifically, the input is a finite ordered sequence of packets that arrive in both directions along a link. 
Given $(u, v)$ and a packet of weight $x$ going from $u$ to $v$, player $u$ can either accept or reject the packet.
If player $u$ accepts the packet, their capacity on link $(u,v)$ decreases by $x$. Correspondingly, player $v$’s capacity on $(u,v)$ increases by $x$. 
If a player rejects the packet, this will entail a cost linear in the weight of the packet. 
A link is ``rechargeable" in the sense that the total capacity of the link has to remain constant, but the allocation of capacity at the ends of the link can depend arbitrarily on players' decisions.
The goal is to minimise the sum of the capacity injected into the link and the 
cost of rejecting packets. 
We show the problem is NP-hard, but can be approximated efficiently with a ratio of $(1+ \eps)\cdot (\arat)$ for some arbitrary $\eps >0$. 
\end{abstract}

\section{Introduction}\label{sec:intro}

This paper considers a novel and natural throughput optimization problem where the goal is to maximise the number of packets routed through a network. The problem variant comes with a twist: link capacities are ``rechargeable'', which is motivated, e.g., by payment-channel networks (more motivation will follow).

We confine ourselves to a single capacitated network link and consider a finite ordered sequence of packet arrivals in both directions along the link.
This can be modelled by a graph that consists of a single edge between two vertices $u$ and $v$, where $b_u$ and $b_v$ represent the capacity $u$ and $v$ injects into the edge respectively. 
Each packet in the sequence has a weight/value and a direction (either going from $u$ to $v$, or from $v$ to $u$). 
When $u$ forwards a packet going in the direction $u$ to $v$, $u$'s capacity $b_u$ decreases by the packet weight and $v$'s capacity $b_v$ correspondingly increases by the packet weight (see \Cref{fig:example} for an example). 
Player $u$ can also reject to forward a packet, incurring a cost linear in the weight of the packet.
The links we consider are rechargeable in the sense that the total capacity $b_u + b_v$ of the link can be arbitrarily distributed on both ends, but the total capacity of the link cannot be altered throughout the lifetime of the link. 
Given a packet sequence, our goal is to minimise the sum of the cost of rejecting packets and the amount of capacity allocated to a link. 

Here we stress a crucial difference between our problem and problems on optimising flows and throughput in typical capacitated communication networks~\cite{chekuri2004all,raghavan1985provably}. 
In traditional communication networks, the capacity is usually independent in the two directions of the link~\cite{Gupta2000TheCO}. 
In our case, however, the amount of packets $u$ sends to $v$ in a link $(u,v)$ directly affects $v$'s capability to send packets, as each packet $u$ send to $v$ increases $v$'s capacity on $(u,v)$.

We start with a description of rechargeable links, then explain the actions players can take and corresponding costs. We subsequently motivate and explain our problem with a real world example of routing payments in cryptocurrency networks. Finally, we state our main results.

\subparagraph{Rechargeable links.} One unique aspect of our problem is that the links we consider are rechargeable. Rechargeable links are links that satisfy the following properties:

\begin{enumerate}
    \item Given a link $(u,v)$ with total capacity $\capacitySum$, the capacity can be arbitrarily split between both ends based on the number and weight of packets processed by $u$ and $v$. 
    That is, $b_u$ and $b_v$ can be arbitrary as long as $b_u + b_v = \capacitySum$.
    \item The total capacity of a link is invariant throughout the lifetime of the link. 
    That is, it is impossible for players to add to or remove any part of the capacity in the link. 
    In particular, if a player is incident to more than one link in the network, the player cannot transfer part of their excess capacity in one link to ``top up" the capacity in another one. 
\end{enumerate}

\begin{figure}[htb!]
    \centering
    \includegraphics[]{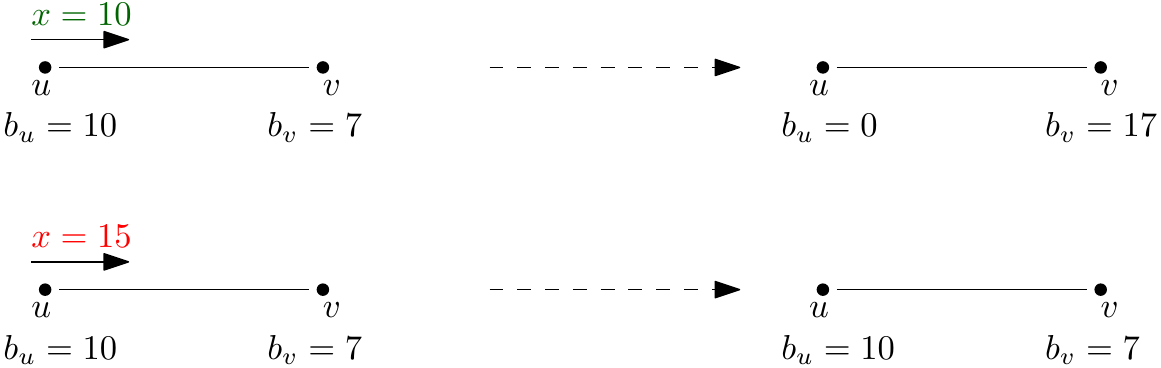}
    \caption{The diagram on the top shows the outcome of $u$ successfully processing a packet $x$ of weight $10$ along the link $(u,v)$. The subsequent capacities of $u$ and $v$ are $0$ and $17$ respectively. The diagram on the bottom shows the outcome where, even though the total capacity of the $(u,v)$ link is $17$, $u$'s capacity of $10$ on $(u,v)$ is insufficient to forward a packet $x$ of weight $15$. As such, the subsequent capacities of $u$ and $v$ on the link $(u,v)$ remain the same.}
    \label{fig:example}
\end{figure}

\subparagraph{Player actions and costs.}
First, we note that creating a link incurs an initial cost of the amount the player allocates in the link. That is, if player $u$ allocates $b_u$ in link $(u,v)$, the cost of creating the link $(u,v)$ for $u$ would be $b_u$. 
Consider a link $(u,v)$ in the network and a packet going from $u$ to $v$ along the edge. Player $u$ can choose to do the following to the packet:
\begin{itemize}
    \item \textbf{Accept packet.} Player $u$ can accept to forward the packet if their capacity in $(u,v)$ is at least the weight of the packet. The result of doing so decreases their capacity by the packet weight and increases the capacity of $v$ by the packet weight. Note that apart from gradually depleting a player's capacity, accepting the packet does not incur any cost. 
    \item \textbf{Reject packet.} Player $u$ can also reject the packet. This could happen if $u$'s capacity is insufficient, or if accepting the packet would incur a larger cost in the future. For a packet of weight $x$, the cost of rejecting the packet is $f \cdot x + m$ where $f,m \in \mathbb{R}^+$.
\end{itemize}
We note that player $u$ does not need to take any action for packets going in the opposite direction (i.e. from $v$ to $u$) as these packets do not add anything to the cost of $u$. See \Cref{sec:definitions} for the formal details of packets.

\subparagraph{Practical motivation.} The primary motivating example of our model is payment channel networks~\cite{decker2015fast,csur21crypto} supporting cryptocurrencies \cite{raiden,poon2015lightning}. These networks are used to route payments of some amount (i.e. weighted packets in our model) between any two users of the network. Channels (i.e. links in our model) are capacitated, which can limit transaction throughput and hence may require selection and recharging \cite{AvarikiotiPSSTY21,KhalilG17}.  Intermediate nodes on a payment route typically also charge a fee for forwarding payments that is linear in the amount of the payment. Hence, if they reject to forward a payment, they would lose out on profiting from this fee and thereby incur the fee amount as opportunity cost. Channels in payment channel networks are also rechargeable for security reasons, see \cite{poon2015lightning} for more details. 

\subparagraph{Our contributions.}
We introduce the natural weighted packet selection problem and show that it is NP-hard by a reduction from subset sum. Our main contribution is an efficient constant-factor approximation algorithm. We further initiate the discussion of how our approach can be generalized from a single link to a more complex network.

\subparagraph{Organisation.}
\Cref{sec:definitions} introduces the requisite notations and definitions we use in our paper. 
\Cref{sec:building} provides the necessary algorithmic building blocks we use to construct our main algorithm. 
In \Cref{sec:approx}, we present our main approximation algorithm and prove that it achieves an approximation ratio of $(1+\eps)(\arat)$ for the \problem problem in \Cref{thm:main}. 
We show that \problem is NP-hard in \Cref{sec:hardness}.
Finally, we provide some generalisations of our algorithm from a single link to a larger network in \Cref{sec:extensions}. 
We conclude our work by discussing future directions in \Cref{sec:future}.

\section{Notation and definitions}\label{sec:definitions}
Let $(u,v)$ be link.
We denote an ordered sequence of packets by $X_t = (x_1, \dots, x_t)$.
Each packet $x_i \in X_t$ has a weight and a direction.
We simply use $x_i \in \mathbb{R}^+$ to denote the weight of the packet $x_i$.
We say a packet $x_i$ goes in the left to right direction (resp. right to left) if it goes from $u$ to $v$ (resp. from $v$ to $u$).

Let $\xlr$ denote the subsequence of $X_t$ that consist of packets going from left to right and $\xrl$ the subsequence of $X_t$ that consist of packets going from right to left.

Let $x_{\min}$ be the weight of the packet with the smallest weight in $X_t$ and $\accAll$ be the capacity needed to accept all packets.
The value of $\accAll$ for $X_t$ is easy to compute in time $\calO(t)$ and it is upper bounded by the sum of the weight of all packets.

Let $\opt$ be the cost of the optimal algorithm and $\opt_{\capacitySum}$ be the cost of the optimal algorithm using a capacity of $\capacitySum$ in the link.
Moreover, we use $\opt^R$ to denote the cost of the optimum for rejecting packets and $\opt^C$ to denote the cost for the capacity.
Similarly, we use $\opt^R_{\capacitySum}$ to denote the cost for rejecting of the optimum using capacity $\capacitySum$ in the link
(note that $\opt^C_{\capacitySum} = \capacitySum$).

Finally, for an integer $t \ge 1$, we use $[t]$ to denote $\{1, \dots, t\}$. 

\section{Preliminary insights and algorithmic building blocks}\label{sec:building}
We start our investigation of the \problem problem by describing a procedure to approximate the optimal capacity in a link.
We describe a linear program that fractionally accepts packets (i.e. part of a packet can be accepted) given a fixed link capacity $\capacitySum$.
Then, we introduce a simple algorithm that requires twice as much link capacity compared to the capacity required in the linear program.
However, the algorithm guarantees that every packet accepted fully (i.e. the entire packet was accepted) by the linear program is also accepted by the algorithm. 
This controls the cost of the algorithm: the fully accepted packets do not increase the cost since the solution of the linear program is a lower bound on the optimal cost.

\subsection{Approximating the optimal capacity}\label{sec:capacity}
We present a lemma that allows us to fix the capacity to $\capacitySum$ for a small trade-off in the approximation ratio.
We fix some $\eps >0$ and perform a search on all $k\in \mathbb{N}$ such that $x_{\min}(1+\eps)^k \le \accAll$.

Observe that if the optimal capacity is not $0$, it is at least $x_{\min}$, the weight of the smallest packet; and at most $\accAll$, the capacity needed to accept all packets.

\begin{lemma}\label{Lemma: Capacity Guess}
For any $\eps > 0$, let $\mathcal{M} = \{x_{\min}(1+\eps)^k | k \in \mathbb{N} \text{ and } x_{\min}(1 + \eps)^k \le \accAll \}\cup\{0\}$. Then, the following inequality holds for any $LB_{\capacitySum} \le \opt_{\capacitySum}^R$
  \[
    \min_{\capacitySum \in \mathcal{M} } \left( LB_{\capacitySum} + \frac{\capacitySum}{1+\eps} \right) \le \opt
  \]
\end{lemma}
\begin{proof}
  If $\opt$ rejected all packets, we know that $LB_{0} \le \opt_{0}^{R} = \opt_0 = \opt$, so the inequality holds.

  Now, suppose that $\opt$ accepted at least one packet.
  That means $\opt^C \ge x_{\min}$.
  So there exists $k$ such that $x_{\min} (1+\eps)^{k-1} \le \opt^C \le x_{\min} (1+\eps)^k$.
  We set $\capacitySum = x_s(1+\eps)^k$ and prove that $LB_{\capacitySum} + \frac{\capacitySum}{1+\eps} \le \opt^R + \opt^C$.
  From the choice of $\capacitySum$, we know that $\frac{M}{1+\eps} \le \opt^C$.

  Observe that for $\capacitySum' \ge \capacitySum$ holds $\opt_{\capacitySum'}^R \le \opt_{\capacitySum}^R$.
  In the worst case, the same set of packets can be accepted with a larger capacity.
  And since $\opt^C \le \capacitySum$, it means $\opt^R \ge \opt_{\capacitySum}^R \ge LB_{\capacitySum}$.
\end{proof}

In \Cref{sec:approx}, we describe an algorithm that is a $(\arat)$-approximation of $LB_{\capacitySum}$.
Together with \Cref{Lemma: Capacity Guess}, we can use this algorithm to approximate the whole problem with a ratio of $(1+\eps)(\arat)$ by running the algorithm at most $\frac{1}{\eps}\log \frac{\accAll}{x_{\min}}$ times.
We note that choosing a smaller value of $\eps$ yields a better approximation, but increases the running time.

\subsection{Linear program formulation}
Here, we describe a linear program that computes a lower bound for $\opt^R_{\capacitySum}$.

Observe that $\opt_{\capacitySum}$ rejects packets with weight larger than $\capacitySum$.
For the rest of the analysis, we assume that all packets in $X_t$ have weight smaller than $ \capacitySum$.

In the linear program, we create a variable $0 \le y_i \le x_i$ for every packet $x_i \in X_t$ that represents the extent to which the packet is accepted ($y_i = \frac{x_i}{2}$ means that half of $x_i$ is accepted).
We introduce variables $S_{L,i}$ and $S_{R,i}$ denoting the capacity on the left and right ends of the link after processing first $i$ packets from $X_t$.
We know that $S_{L,i} + S_{R,i} = \capacitySum$, and $ 0\le S_{L,i}, S_{R,i} \le \capacitySum$.

Now we can formulate the linear program in eq.~\eqref{eq:big}:
\begin{align}\label{eq:big}
  &\text{minimise} \quad \sum_{i} f(x_i - y_i) + m\frac{x_i - y_i}{x_i}\\
  &\text{subject to} \quad \forall i : y_i, S_{L,i},S_{R,i} \ge 0 \notag\\
                         &\quad\quad\quad\quad\quad  \forall i : y_i \le x_i \notag \\
                         &\quad\quad\quad\quad\quad  \forall i : S_{L,i} + S_{R,i} = M \notag\\
                         &\quad\quad\quad\quad\quad  \forall x_i \in \xlr: S_{L,i} = S_{L,i-1} - y_i \notag\\
                         &\quad\quad\quad\quad\quad  \forall x_i \in \xlr: S_{R,i} = S_{R,i-1} + y_i \notag\\
                         &\quad\quad\quad\quad\quad  \forall x_i \in \xrl: S_{L,i} = S_{L,i-1} + y_i \notag\\
                         &\quad\quad\quad\quad\quad  \forall x_i \in \xrl: S_{R,i} = S_{R,i-1} - y_i \notag
\end{align}

Let $LP_{\capacitySum}$ be the solution of the linear program with capacity parameter $\capacitySum$. \Cref{lem:lb} states that $LP_{\capacitySum}$ is a lower bound of the optimal cost of the \problem problem with capacity $\capacitySum$.
\begin{lemma}\label{lem:lb}
  For all $\capacitySum$, $\opt_{\capacitySum} \ge LP_{\capacitySum}$.
\end{lemma}
\begin{proof}
  The solution of $\opt_{\capacitySum}$ is an admissible solution to the linear program.
  If some other (fractional) solution is found, we know that it is at most $\opt_{\capacitySum}$.
\end{proof}

The linear program can be solved in time $\mathcal{O}(n^{\omega})$ where $n$ is the number of variables in the linear program and $\omega$ the matrix multiplication exponent~\cite{CohenLS21} (currently $\omega$ is around $2.37$).

\subsection{Processing fully accepted packets}

Given the solution of the linear program with capacity $\capacitySum$, we describe an algorithm that: uses capacity $2\capacitySum$; accepts all packets that were fully accepted ($x_i = y_i$) by the linear program; and accepts some fractionally accepted packets.

\Cref{Algorithm: Accept Full} describes the decision making process only for packets coming from left to right on the link, i.e. $\xlr$.
The algorithm takes as input the solution to the linear program and the packet sequence $X_t$.
Recall that  $S_{L,i}$ and $S_{R,i}$ for $i \in [t]$ are the capacity distributions on the left and right end of the link after processing the $i$th packet from the linear program solution.
First, we set $R_{L}$ and $R_{R}$ be $\frac{M}{2}$ each and we start with $S_{L,0} + R_{L}$ capacity on the left and $S_{R,0}+ R_{R}$ capacity on the right of the link. 
Intuitively, one can think of the additional $\capacitySum$ capacity in $R_{L}$ and $R_{R}$ as an excess capacity to accept fractionally accepted packets or balance change in $S_{L,i}$ and $S_{R,i}$ in the case of packet rejection.
We stress that the algorithm always uses the capacity in $S_{L,i}$ and $S_{R,i}$ to accept the fractional portion of the $i$th packet. 
The remainder capacity comes from $R_{L}$ and $R_{R}$.

Since the problem is symmetric, one can simply swap $S_{L,i}$ for $S_{R,i}$, $\xlr$ for $\xrl$, and $R_L$ for $R_R$ in \Cref{Algorithm: Accept Full} to get decisions for packets from $\xrl$.

\begin{algorithm}[t]
  \begin{algorithmic}[1]
    \Require packet sequence $X_t$, capacity $\capacitySum$, solution of LP: $S_{L,i}, S_{R,i}, y_i$.
    \Ensure decisions to accept or reject
    \State $R_{L} = \frac{\capacitySum}{2}$, $R_{R} = \frac{\capacitySum}{2}$
    \For{ $i \in [t]$}
     \If {$x_i \in \xlr$}
      \If {$R_L \ge x_i - y_i$}
        \State \textbf{Accept}
        \State $R_L = R_L - (x_i - y_i)$
        \State $R_R = R_R + (x_i - y_i)$
      \Else
        \State \textbf{Reject}
        \State $R_L = R_L + y_i$
        \State $R_R = R_R - y_i$
      \EndIf
     \EndIf
    \EndFor
	\end{algorithmic}
  \caption{Algorithm accepting all fully accepted packets}
  \label{Algorithm: Accept Full}
\end{algorithm}

\begin{lemma}\label{Lemma: Accept Full}
  Given the solution of the linear program, a sequence of packets $X_t$, and link capacity $\capacitySum$,
  \Cref{Algorithm: Accept Full} uses capacity $2\capacitySum$ and accepts all packets fully accepted in the linear program.
\end{lemma}
\begin{proof}
  After processing the $i$th packet $x_i \in \xlr$, we denote $R_L$ at that step $R_{L,i}$, similarly with $R_{R,i}$.
  We show that the channel, at time $i$, has capacity at least $S_{L,i}$ on the left and at least $S_{R,i}$ on the right.

  When $R_{L,i}$ is large enough to accept packet $x_i$, we use $y_i$ capacity from $S_{L,i}$ and $x_i - y_i$ capacity from $R_{L,i}$.
  The capacity $y_i$ from the accepted packet goes to $S_{R,i+1}$ and the rest ($x_i - y_i$) of the capacity goes to $R_{R,i+1}$.

  If the packet is forced to be rejected, we know that $R_{L,i} < x_i - y_i$.
  Since $R_{R,i} = M - R_{L,i}$, we know that $R_{R,i} > M - x_i +y_i$, and because all packets have weight smaller than $\capacitySum$, $R_{R,i} > y_i$ follows.
  This means we can take $y_i$ from $R_{R,i}$ and put add it to $S_{R,i+1}$ and remove $y_i$ from $S_{L,i}$ (because the capacity disappeared from there) to $R_{L,i+1}$.

  If the packet is fully accepted, then $x_i - y_i = 0$. Than means the condition $R_{L,i} \ge x_i - y_i$ is satisfied and the algorithm accepts it.
\end{proof}

We conclude this section with two remarks.
\begin{remark}
\Cref{Lemma: Accept Full} holds for any initial distribution of $R_{L,0}$ and $R_{R,0}$ so long as $R_{L,0} + R_{R,0} = \capacitySum$.
\end{remark}

\begin{remark}
\Cref{Algorithm: Accept Full} is greedy and accepts all packets as long as 
$R_L \ge x_i-y_i$ (which could be suboptimal as it might not have enough capacity in $R_L$ to accept important packets later in the sequence). However, to maintain the condition $R_L, R_R \ge 0$ in line $4$ of \Cref{Algorithm: Accept Full}, we can substitute the conditional check $R_L \ge x_i-y_i$ with $R_R < y_i$ at any point. Then proof of \Cref{Lemma: Accept Full} still holds.
We note that one could use this as a heuristic to develop a better approximation as it gives some control over how greedy the algorithm is.
\end{remark}

\section{A constant approximation algorithm}\label{sec:approx}

Based on the insights above, we present a $(\arat)$-approximation algorithm for \problem with fixed capacity $\capacitySum$.
The algorithm modifies \Cref{Algorithm: Accept Full} by adding capacity to $R_L$ and $R_R$, which allows us to make a wider range of decisions.

Let \emph{little-accepted} (packets) be packets for which $\frac{y_i}{x_i} < \frac{\sqrt{3}}{\arat}$ holds,
and \emph{almost-accepted} be packets for which $\frac{y_i}{x_i} \ge  \frac{\sqrt{3}}{\arat}$ holds.

The algorithm keeps the capacity reserve in $R_L$ and $R_R$ high while accepting all almost-accepted packets.
It might not always have enough capacity to do so.
We present procedures that deal with that situation.

In \Cref{Algorithm: Appx}, we again describe the decision making process for packets $x_i \in \xlr$. The decisions for packets from $\xrl$ are symmetric.

The following lemma states that we can safely reject all little-accepted packets.

\begin{lemma}\label{Lemma: Little Accepted}
  Observe that all little-accepted packets can be rejected while keeping the approximation ratio below $\arat$.
\end{lemma}
\begin{proof}
Recall that rejecting a packet $x_i$ incurs a cost of $fx_i + m$.
If $x_i$ is little-accepted, the cost incurred by the linear program solution is
$f\cdot(x_i - y_i) + m\frac{x_i - y_i}{x_i} \ge \frac{f x_i}{\arat} +\frac{m}{\arat} = \frac{1}{\arat} (fx_i + m)$.
\end{proof}

\begin{algorithm}[t]
  \begin{algorithmic}[1]
    \Require packet sequence $X_t$, capacity $\capacitySum$, solution of LP:$S_{L,i}, S_{R,i}, y_i$.
    \Ensure decisions to accept or reject
    \State $R_{L} = \capacitySum\frac{1 + \hibu}{2}$, $R_{R} = \capacitySum\frac{1 + \hibu}{2}$
    \For{$i \in [t]$}
     \If {$x_i \in \xlr$}
      \If{$R_L - (x_i - y_i) \ge \frac{\hibu}{2}$}
        \State \textbf{Accept}
        \State $R_L = R_L - (x_i-y_i)$
        \State $R_R = R_R + (x_i-y_i)$
      \ElsIf{$x_i$ is little-accepted}
        \State \textbf{Reject}
        \State $R_L = R_L + y_i$
        \State $R_R = R_R - y_i$
      \Else
        \State $\surac, \surre, \undecided, R_L', j \gets \Call{Divide}{R_L,LP,X_t, i}$ \label{line: divide}
        \State $\undecided_R \gets \{\}$
        \If {$R_L' < 0$}
          \State $\undecided_R, R_L' \gets \Call{RejectBig}{X_t, \undecided, R_L'}$ \label{line: reject}
        \EndIf
        \State \textbf{Accept} all $x_i \in \surac \cup (\undecided \setminus \undecided_R)$ \label{line: accept}
        \State \textbf{Reject} all $x_i \in \surre \cup \undecided_R$.
        \State $R_L = R_L'$
        \State $R_R = (1 + \hibu)M - R_L$
        \State $i = j$
      \EndIf
      \EndIf
    \EndFor
  \end{algorithmic}
  \caption{$(\arat)$-approximation algorithm}
  \label{Algorithm: Appx}
\end{algorithm}

We distinguish between three phases of the algorithm.
The algorithm is in the \emph{balanced phase} if both $R_L \ge \frac{\hibu}{2}$ and $R_R \ge \frac{\hibu}{2}$.
If $R_L < \frac{\hibu}{2}$, we say the algorithm is in the \emph{left phase}, and if $R_R < \frac{\hibu}{2}$, we say the algorithm is in the \emph{right phase}.

In the balanced phase, \Cref{Algorithm: Appx} accepts all almost-accepted packets and those little-accepted packets that allows it to stay in the balanced phase.
When the algorithm is forced to leave the balanced phase and enters the left-phase (or right-phase), it looks at future packets and tries to accept all almost-accepted packets.
If this is not possible, the algorithm rejects some of them such that both of the following two conditions hold: first, the approximation ratio remains $\arat$, and second, the algorithm returns to a balanced phase.
Right and left phases are handled by the functions \Call{Divide}{} (described in \Cref{Algorithm: Divide}) and \Call{RejectBig}{} (described in \Cref{Algorithm: Reject Big}).

\begin{lemma}\label{Lemma: Leaving}
  \Cref{Algorithm: Appx} never leaves the balanced phase after processing a little-accepted packet.
\end{lemma}
\begin{proof}
  Each little-accepted packet moves at most $\frac{1}{\arat} \capacitySum$ from the left side to the right side of a link,
  and at most $\frac{\sqrt{3}}{\arat} \capacitySum$ from the right side to the left side of a link.

  Because $R_L + R_R = \sqrt{3}$ and any packet has a weight at most $\capacitySum$.
  If $R_L - \frac{1}{\arat}\capacitySum < \frac{\hibu}{2}$, then $R_R - \frac{\sqrt{3}}{\arat}\capacitySum \ge \frac{\hibu}{2}$.
  
  That means that rejecting a little-accepted packet from $\xlr$ does not create a situation where $R_R < \frac{\hibu}{2}$.
\end{proof}

We have no guarantee that \Cref{Algorithm: Appx} accepts all almost-accepted packets.
Algorithms \Call{Divide}{} and \Call{RejectBig}{} manage capacity when by accepting almost-accepted packets from $\xlr$ leads to $R_L < \frac{\hibu}{2}$.

First, \Call{Divide}{} creates three sets from some future packets: $\surac, \surre, \undecided$.
Set $\surac$ contains all packets from $\xrl$, these will be accepted.
Set $\surre$ contains little-accepted packets from $\xlr$, these will be rejected.
Set $\undecided$ contains almost-accepted packets, some of them will be accepted and some rejected in a way to maintain the approximation ratio.

\Call{Divide}{} creates the sets incrementally.
It simulates accepting packets from $\surac \cup \surre$ and rejecting packets from $\surre$
until one of the following stopping conditions occurs:
\begin{itemize}
  \item the algorithm runs out of capacity ($R_L$ would be smaller than $0$)
  \item the algorithm returns to a balanced phase ($R_L$ would be bigger than $\frac{\hibu}{2}$).
  \item all packets are processed
\end{itemize}

If the first condition holds, the procedure \Call{RejectBig}{} is called and the procedure creates set $\undecided_R \subset \undecided$.
Packets in $\surac$ and $\undecided\setminus \undecided_R$ are accepted and packets in $\surre$ and $\undecided_R$ are rejected.

\begin{lemma}\label{Lemma: Positive}
  If \Call{Divide}{} returns $R_L' \ge 0$ and $j$, all almost-accepted packets between $i$ and $j$ are accepted
  by \Cref{Algorithm: Appx} and either all packets are processed or $R_{R,j} \ge \frac{\hibu}{2}$ and $R_{L,j} \ge \frac{\hibu}{2}$.
\end{lemma}
\begin{proof}
  There are two reasons why \Call{Divide}{} returned $R_L' \ge 0$, either $R_L' \ge \frac{\hibu}{2}$ or $j = t$.

  In both cases \Call{Divide}{} simulated accepting all packets from $\surac$ and $\undecided$ and rejecting all packets from $\surre$, and at no time $R_L'$ went below $0$.
  That means that \Cref{Algorithm: Appx} just repeats decisions of \Call{Divide}{}.
\end{proof}

\begin{algorithm}[t]
  \caption{Function \Call{Divide}{} to create sets $\surac, \surre$, and $\undecided$.}
  \label{Algorithm: Divide}
  \begin{algorithmic}[1]
    \Require packet sequence $X_t$, solution of LP :$ S_{L,i}, S_{R,i}, y_i$, value $R_L$, capacity $\capacitySum$, 
    \Ensure sets $\surac, \surre, \undecided$, resulting $R_L$
    \State $R_L = R_L - (x_i - y_i)$
    \State $\surac, \surre, \undecided \gets \{\}, \{\}, \{x_i\}$
    \State $j = i$
    \While{$R_L \ge 0$ and $R_L < \frac{\hibu}{2}$ and $j < t$}
      \State $j = j+1$
      \If{ $x_j \in \xlr$ and $x_j$ is almost-accepted}
        \State $R_L = R_L - (x_j - y_j)$
        \State $\undecided \gets \undecided \cup {x_j}$
      \ElsIf{ $x_j \in \xlr$ and $x_j$ is little-accepted}
        \State $R_L = R_L + y_j$
        \State $\surre \gets \surre \cup {x_j}$
      \Else
        \State $R_L = R_L + (x_j - y_j)$
        \State $\surac \gets \surac \cup {x_j}$
      \EndIf
    \EndWhile
    \State \Return $\surac, \surre, \undecided, R_L, j$
  \end{algorithmic}
\end{algorithm}

Note also that $R_L + R_R = \sqrt{3}$ and all packets are smaller than $\capacitySum$.
That means that \Cref{Algorithm: Appx} after emerging from left-phase cannot plunge to a right-phase right away. 

\begin{algorithm}[t]
  \caption{Function \Call{RejectBig} to prune out packets from $\undecided$.}
  \label{Algorithm: Reject Big}
  \begin{algorithmic}[1]
    \Require packet sequence $X_t$, set $\undecided$, value $R_L'$
    \Ensure set $\undecided_R$, value $R_L'$
    \State $\undecided_R \gets \{\}$
    \While{ $R_L' < \frac{\hibu}{2}$}
      \State $x_k \gets $ biggest packet from $\undecided$
      \State $\undecided \gets \undecided \setminus x_k$
      \State $\undecided_R \gets \undecided_R \cup \{x_k\}$
      \State $R_L' = R_L' + x_k$
    \EndWhile
    \State \Return $\undecided_R, R_L'$
  \end{algorithmic}
\end{algorithm}

To prove that \Call{RejectBig}{} maintains the approximation ratio, we compute the cost incurred by the algorithm on packets from $\undecided_R$ and compare it to the cost of $\opt$ on $\undecided$.

\begin{lemma}\label{Lemma: Appx}
  In \Cref{Algorithm: Appx}, for sets $\undecided$ and $\undecided_R$ on \cref{line: accept} holds
  \[
    (\arat) \sum_{x_i \in \undecided} f \cdot \left( x_i - y_i \right) + m  \frac{x_i - y_i}{x_i} \ge  \sum_{x_i \in \undecided_R} f x_i + m
  \]
\end{lemma}
\begin{proof}
  If $R_L' \ge 0$, we know that all almost-accepted packets are accepted from \Cref{Lemma: Positive}.
  For $R_L' < 0$ we prove that $(\arat) \sum_{x_i \in U}  \left( x_i - y_i \right) \ge  \sum_{x_i \in \undecided_R} x_i$, then we argue that the whole theorem holds.

  Let $D = R_{L,i-1} - R_L'$ where $R_L'$ is the value returned by \Call{Divide}{} in \Cref{Algorithm: Appx} on \cref{line: divide}.
  We know that $D \ge \frac{\hibu}{2} \capacitySum$.

  By following the changes of $R_L'$ in \Call{Divide}{}, we get
  \[
    \sum_{x_i \in \undecided} x_i - y_i = D + \sum_{x_i \in \surre} y_i  + \sum_{x_i \in \surac} x_i - y_i
  \]
  By this we know that $\sum_{x_i \in \undecided} x_i-y_i \ge D$.

  Algorithm \Call{RejectBig}{} removes packets from $\undecided$ until $\sum_{x_i \in \undecided_R} x_i \ge D$.
  If the condition is satisfied, we know that \Call{RejectBig}{} returns $\undecided_R$, because $R_L' \ge \frac{\hibu}{2}$.

  If $|\undecided_R| = 1$, we know that $\sum_{x_i \in \undecided_R } x_i \le \capacitySum$, because every $x_i \le \capacitySum$.
  So in that case $\sum_{x_i \in \undecided_R} x_i \le \capacitySum \le (\arat) \frac{\hibu}{2} \capacitySum \le (\arat)D$.

  If $|\undecided \setminus \undecided_R| > 1$, we know that rejecting just one packet is not enough.
  This means the biggest packet has weight at most $D$, so $\sum_{x_i \in \undecided_R} x_i \le 2D \le (\arat)D$.

  Now, we know that \Cref{Algorithm: Appx} rejects less weight than the linear program times $(\arat)$.
  It implies that $(\arat) \sum_{x_i \in U} f \cdot \left( x_i - y_i \right)\ge  \sum_{x_i \in \undecided_R} f x_i$
  and leaves us to prove $(\arat) \sum_{x_i \in U} \frac{x_i - y_i}{x_i} \ge  \sum_{x_i \in \undecided_R} m$.

  But we know that the packets are moved to $\undecided_R$ from the biggest.
  For every $x_k \in \undecided_R$ and $x_l \in \undecided$ holds $\frac{x_l - y_l}{x_l} \ge \frac{x_l - y_l}{x_k}$.
  That means rejecting smaller packets incurrs on average bigger cost than rejecting bigger packets, so $(\arat) \sum_{x_i \in U} \frac{x_i - y_i}{x_i} \ge  \sum_{x_i \in \undecided_R} m$.
\end{proof}

We now have all the necessary ingredients to state and prove our main theorem. 

\begin{theorem}\label{thm:main}
  \Problem can be approximated with a ratio $(1 + \eps)(\arat)$ in time $\calO(n^{\omega}\cdot \frac{1}{\eps} \cdot \log \frac{\accAll}{x_{\min}})$,
  where $\omega$ is the exponent of $n$ in matrix multiplication.
\end{theorem}
\begin{proof}
  We estimate the capacity according to \Cref{Lemma: Capacity Guess} and for every estimate we solve the linear program~\eqref{eq:big} and run \Cref{Algorithm: Appx}.
  The solution is the output of \Cref{Algorithm: Appx} with the smallest cost.

  We know that $x_{\min}(1 + \eps)^{\frac{1}{\eps}\cdot \log \frac{\accAll}{x_{\min}}} \ge \accAll$.
  That means we solve the linear program and run \Cref{Algorithm: Appx} at most $\frac{1}{\eps} \cdot \log \frac{\accAll}{x_{\min}}$ times.

  From \Cref{lem:lb} we know that the solution of the linear program with parameter $\capacitySum$ is a lower bound for the $\opt_{\capacitySum}^R$.

  From \cref{Lemma: Accept Full}, we know that \Cref{Algorithm: Appx} accepts all fully-accepted packets.
  The algorithm can reject any little-accepted packets by \Cref{Lemma: Little Accepted}. In balanced phase it accepts all almost-accepted packets
  and never leaves the phase after seeing little-accepted packets (\Cref{Lemma: Leaving}).
  \Cref{Lemma: Appx} says even in a left (or right) phase the approximation ratio on almost-accepted packets is $\arat$.
  This means \Cref{Algorithm: Appx} is $(\arat)$-approximation algorithm for the solution of the linear program.
  Moreover, the algorithm uses $\arat$ times more capacity that the linear program.

  Using \Cref{Lemma: Capacity Guess}, we find that the selected solution is a $(1 + \eps)(\arat)$-approximation of the \problem problem.
\end{proof}

\section{Hardness}\label{sec:hardness}
In this section, we show that \problem is generally NP-hard.
\begin{theorem}
  \Problem is NP-hard.
\end{theorem}
\begin{proof}
  We show a reduction from the subset sum problem, which is known to be NP-hard~\cite{complexitybook}. 
  In the subset sum problem, we are given a multiset of integers $\mathcal{I} := \{i_1,i_2,\dots, i_n\}$ and a target integer $S$.
  The goal is to find a subset of $\mathcal{I}$ with a sum of $S$.
  
  Consider the following question in the \problem problem: "is the cost below a given value?" 
  We show this question is NP-hard.

  We set the constants to $m = 0$ and $f = \frac{3}{4}$.
  We create a packet sequence consisting of $i_1, i_2, \dots, i_n$ where the $j$th packet in the sequence has weight $i_j$ which is the $j$th element in $\mathcal{I}$.
  These packets all go from left to right.
  Then we add a packet of weight $S$ going from right to left.

  Suppose that there exists $\mathcal{I}' \subseteq \mathcal{I}$, such that $\sum_{j \in \mathcal{I}'} i_j = S$.
  Then we show that the cost is at most $\frac{1}{4}S + \frac{3}{4} \sum_{j \in \mathcal{I}} i_j$.

  The solution reaching that cost is as follows: players start with capacity $S$ on the right and accept all packets from $\mathcal{I}'$ and then accept the last packet of weight $S$.
  The cost is then $S + \frac{3}{4} \sum_{j \in \mathcal{I} \setminus \mathcal{I}'} i_j$.
  Since $\sum_{j \in \mathcal{I}'} i_j = S$, the bound holds.

  Now, suppose that there is no subset of $\mathcal{I}$ summing to $S$.
  Let $\mathcal{A} \subseteq \mathcal{I}$ be any set with sum $A$
  The cost for the packets going from left to right is $A + \frac{3}{4} \sum_{j \in \mathcal{I} \setminus \mathcal{A}} i_j = \frac{1}{4}A + \frac{3}{4} \sum_{j \in \mathcal{I}} i_j$.
  Depending whether the last packet was accepted or rejected; or $A < S$ or $A > S$, we need to add $\min(\max(S-A, 0), \frac{3}{4}S)$.
  Since $A \neq S$, we know that 
  \begin{align*}
    \frac{1}{4}A + \frac{3}{4} \sum_{j \in \mathcal{I}} i_j + \min(\max(S-A, 0), \frac{3}{4}S) &> \frac{1}{4}S + \frac{3}{4} \sum_{j \in \mathcal{I}} i_j\\
    \min(\max(S-A, 0), \frac{3}{4}S) &> \frac{1}{4}(S-A)\\
  \end{align*}
  which means that \problem is NP-hard.
\end{proof}

\section{Extensions}\label{sec:extensions}
We highlight two natural and interesting directions to generalise our approach from a link to a network. 

\subsection{Cyclic redistribution of capacity to reduce cost}
Suppose player $u$ on link $(u,v)$ is incident to $\ge 2$ links (let us call one of the incident links $(u,w)$). 
From our definition of rechargeable links (see \Cref{sec:intro}), we know it is not possible for $u$ to increase the capacity on the $(u,v)$ link by transferring excess capacity from $(u,w)$. 
However, if $(u,v)$ and $(u,w)$ are part of a larger cycle in the network, $u$ can send excess capacity from link to link in a cyclic fashion starting from $(u,w)$ link and ending at $(u,v)$ while maintaining the invariant that the total capacity on each link as well as the sum of all the capacities of a player on their incident links remains the same. 
This can be done at any point in time without the need to transfer packets. 
We call this cyclic redistribution (note this is possible in payment channel networks \cite{PickhardtN20, KhalilG17, AvarikiotiPSSTY21}) and illustrate it with an example in \Cref{fig:offchain}.
In some situations, especially if the cost of destroying and recreating a link is extremely large, the possibility of cheaply shifting capacities in cycles can reduce the overall cost of the algorithm.

Let us denote the cost of decreasing capacities by $x$ on the right and increasing it by $x$ on the left using cyclic redistribution by $C(fx +m)$ for some $C \ge 1$ (one can view $C$ as a function of the length of the cycle one sends the capacities along).

Here, we sketch an approximation algorithm that solves the \problem problem with the possibility of cyclic redistribution.
Note that our sketch is not precise, we simply modify \Cref{Algorithm: Appx} where constants are already optimised for the basic problem.

We modify the linear program by adding variables $o_i, i \in [t]$ with constraints $0 \le o_i \le \capacitySum$.
The variable $o_i$ denotes the capacity that was shifted from one side to the other before the algorithm processes packet $x_i$.
We also modify the capacity constraints in the following way (for the case where $x_{i-1} \in \xrl$ and $x_{i} \in \xrl$): $S_{L, i} =  S_{L,i-1} - o_i + y_{i-1}$ and $S_{R,i} = S_{R,i-1} + o_i - y_{i-1}$.
We change the signs of variables for the other cases. 
Finally, we add $\sum_{i} C(fo_i + \frac{o_i}{M} m)$ to the objective in \Cref{eq:big}.

We divide the algorithm into epochs. 
We sum all $o_i$ in the current epoch.
If the sum is above $\frac{1}{1+\sqrt{3}}\capacitySum$ we perform cyclic distribution if needed and start a new epoch.
Note that in the current epoch, $\opt$ already paid at least $Cf\frac{\capacitySum+m}{1+\sqrt{3}}$,
so, we can move $\capacitySum$ capacity, incurring a cost at most $1+\sqrt{3}$ times bigger than $\opt$ for cyclic redistribution.

To deal with capacity changes inside each epoch, we increase $R_L$ and $R_R$.
We initialise them in a way that they absorb changes of capacity in the first epoch of our algorithm.
After an epoch, we reset them by cyclic redistribution such that they absorb changes of capacity in the next.
The increase in $R_L$ and $R_R$ is at most $\frac{1}{1+\sqrt{3}}\capacitySum$.
These changes increase the approximation ratio of our algorithm from $\arat$ to $\arat + \frac{1}{\arat}$, which is $\frac{1 + 3 \sqrt{3}}{2}$.

\begin{figure}[htb!]
    \centering
    \includegraphics[]{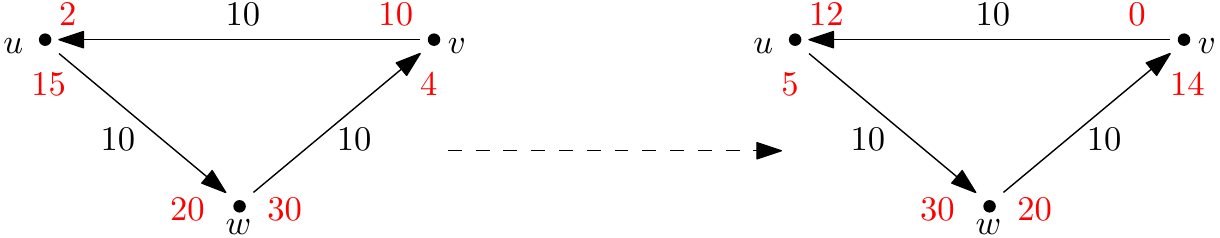}
    \caption{The graph on the left depicts three players $u,v,w$ connected in a cycle. The red numbers by each link represent the capacity of a player in a certain link. Player $u$ can increase their capacity by $10$ on the $(u,v)$ link by first sending the excess capacity of $10$ to $w$ along the $(u,w)$ link. Then player $w$ sends the excess capacity of $10$ to $v$ along $(w,v)$. Finally, player $v$ sends capacity of $10$ back to $u$ along $(v,u)$. The graph on the right depicts the updated capacities of each player on each link after cyclic redistribution.}
    \label{fig:offchain}
\end{figure}

\subsection{Graphs with a few long packet requests}

We can extend weighted packet selection from a single link to a network.
We first modify the definition of a packet such that, in addition to having a non-negative weight, it now has a sequence of (directed) links it needs to be processed by. 
These links form the path that packet takes through the network. 
The packet has to be be accepted or rejected by \emph{all} the links in its routing path.
We call a packet \emph{long} if it passes through more than one link.

Thus far, we showed an approximation algorithm that solves the problem if all packets only go through $1$ link.
Suppose we are given the situation where we can bound the number of long packets, say by $\ell$.
Given a network and packet sequence for which we know only $\ell$ are long, we can approximate the extended problem with approximation ratio $(1 + \eps)(\arat)$
in time $2^{\ell}$ times the time needed for one link by simply trying all to accept all subsets of paths of long packets.
If the packet is accepted or rejected, we can reflect it in the linear program by requiring $y_i = x_i$.
Then \Cref{Algorithm: Appx} surely accepts this packet and the condition that the packet needs to be accepted by all links it passes through is satisfied.

\section{Discussion and future work}\label{sec:future}

We initiated the study of weighted packet selection over a rechargable capacitated link, a natural  algorithmic problem e.g.,
describing the routing of financial transactions in cryptocurrency networks.
We showed that this problem is NP-hard and provided a constant factor approximation algorithm.

We understand our work as a first step, and believe that it opens several interesting avenues for future research.
In particular, it remains to find a matching lower bound for the achievable approximation ratio, and to study
the performance of our algorithm in practice.
More generally, it would be interesting to study the online version of the weighted packet selection problem,
and explore competitive algorithms. This version of the problem, when extended to a network, can be seen as a novel 
version of the classic online call admission problem~\cite{aspnes1997line}.

\noindent \textbf{Acknowledgments.}
We thank Mahsa Bastankhah and Mohammad Ali Maddah-Ali for fruitful discussions about different variants of the problem.

\bibliography{refs}

\end{document}